\documentclass[10pt]{article}
\usepackage{tikz}
\usepackage{enumitem}
\usepackage{latexsym,verbatim}
\usepackage{url,color}
\usepackage{amsmath}
\usepackage{amssymb,pict2e}

\newcommand{\Z}{\mathbb{Z}}
\newcommand{\abelian}{{\mathcal{G}}}
\newcommand{\df}{\mathrm{DDF}}
\newcommand{\extl}{{\mathcal{E}}}
\newcommand{\intl}{{\mathcal{I}}}
\newcommand{\orb}{{\mathcal{O}}}
\newcommand{\s}{{\mathbf s}}
\newcommand{\w}{{\mathbf w}}
\newcommand{\fhs}{{\mathcal{F}}}

\newcommand{\packing}{{\mathcal{P}}}
\newcommand{\hyperpl}{{\mathcal{H}}}
\newcommand{\group}{\langle \tau \rangle}

\newcommand{\gf}[1]{\ensuremath{\mathrm{GF}(#1)}}
\newcommand{\rot}[1]{\ensuremath{\overset{\Leftrightarrow}{#1}}}

\newtheorem{definition}{Definition}[section]
\newtheorem{lem}[definition]{Lemma}
\newtheorem{theorem}[definition]{Theorem}
\newtheorem{example}[definition]{Example}
\newtheorem{rem}[definition]{Remark}

\newenvironment{remark}{\begin{rem} \em}{\hfill $\Box$ \end{rem}}

\newenvironment{lemma}{\begin{lem} \em}{\hfill $\Box$ \end{lem}}

\newenvironment{proof}{{\bf Proof:} }{\hfill $\Box$}  

\begin{document}
\title{Disjoint difference families and their
  applications} \author{S.~-L.~Ng\footnote{Information Security Group, Royal Holloway University of London, Egham, Surrey TW20 0EX, United Kingdom. {\tt s.ng@rhul.ac.uk}},  M.~B.~Paterson\footnote{Department of Economics, Mathematics and Statistics, Birkbeck, University of London, Malet Street, Bloomsbury,
London WC1E 7HX, United Kingdom. {\tt m.paterson@bbk.ac.uk}}} 
\date{October 8, 2015}
\maketitle

\begin{abstract}
  Difference sets and their generalisations to difference
  families arise from the study of designs and many other
  applications.  Here we give a brief survey of some of these
  applications, noting in particular the diverse definitions of
  difference families and the variations in priorities in
  constructions.  We propose a definition of disjoint difference
  families that encompasses these variations and allows a comparison of
  the similarities and disparities.  We then focus on two
  constructions of disjoint difference families arising from frequency
  hopping sequences and show that they are in fact the same.
  We conclude with a discussion of the notion of equivalence for
  frequency hopping sequences and for disjoint difference families.

{\bf Keywords:} {Frequency hopping sequences, difference families, m-sequences, finite geometry.}
 
{\bf Classification:} {94C30, 51E20, 94A62, 05B10.}

\end{abstract}

\section{Introduction}

Difference sets and their generalisations to difference families arise
from the study of designs and many other applications.  In particular,
the generalisation of difference sets to internal and external
difference families arises from many applications in communications and
information security.  Roughly speaking, a difference family consists
of a collection of subsets of an abelian group, and internal
differences are the differences between elements of the same subsets,
while external differences are the differences between elements of
distinct subsets.  Most of the definitions do not coincide exactly
with each other, understandably since they arise from diverse
applications, and the priorities of maximising or minimising various
parameters are also understandably divergent.  However, there is
enough overlap in these definitions to warrant a study of how they
relate to each other, and how the construction of one family may
inform the construction of another.  One of the aims of this paper is
to perform a brief survey of these difference families, noting the
variations in definitions and priorities, and to propose a definition
that encompasses these definitions and allows a more unified study of
these objects.

One particular class of internal difference family arises from
frequency hopping (FH) sequences.  FH sequences allow many transmitters to
send messages simultaneously using a limited number of channels and it
transpires that the question of how efficiently one can send messages
has to do with the number of internal differences in a collection of
subsets of frequency channels. The seminal paper of Lempel and
Greenberger \cite{lempelgreenberger} gave optimal FH sequences using
transformations of linear feedback shift register (LFSR) sequences.  In
another paper by Fuji-Hara {\it et al.} \cite{FHS} various families of
FH sequences were constructed using designs with particular automorphisms, and
the question was raised there as to whether these constructions are
the same as the LFSR constructions in \cite{lempelgreenberger}.  Here
we show a correspondence between one particular family of
constructions in \cite{FHS} and that of \cite{lempelgreenberger}.

The relationship between the equivalence of difference families and
the equivalence of the designs and codes that arise from them has been
much studied.  Here we will focus on the notion of equivalence for
frequency hopping sequences and for disjoint difference families. 
%and we will point to further open questions in this area.

\subsection{Definitions}

Let $\abelian$ be an abelian group\footnote{The Handbook of
  Combinatorial Designs \cite{handbook} has more material on
  difference families defined on non-abelian groups, but we will focus
  on abelian groups here since most of the applications we examine
  use abelian groups.} of size $v$, and let $Q_0, \ldots,
Q_{q-1}$ be disjoint subsets of $\abelian$, $|Q_i|=k_i$, $i=0, \ldots,
q-1$.  We will call $(\abelian; Q_0, \ldots, Q_{q-1})$ a
\emph{disjoint difference family} $\df(v; k_0, \ldots, k_{q-1})$ over
$\abelian$ with the following \emph{external} $\extl(\cdot)$ and
\emph{internal} $\intl(\cdot)$ \emph{differences}:

\begin{eqnarray*}
\extl_{i,j}(d) & = & \{ (a,b) \; : \; a -b = d, a \in Q_i, b \in Q_j,j\neq i\},\\ 
\extl_{i}(d) & = & \{ (a,b) \; : \; a -b = d, a \in Q_i, b \in Q_j, 
                      j = 0, \ldots, q-1, j \neq i\}, \\ 
\extl(d) & = & \{ (a,b) \; : \; a -b = d, a \in Q_i, b \in Q_j, 
                    i, j = 0, \ldots, q-1, i \neq j \}, \\
\intl_{i}(d) & = & \{ (a,b) \; : \; a -b = d, a, b \in Q_i, a \neq b\},\\
\intl(d) & = & \{ (a,b) \; : \; a -b = d, a, b \in Q_i, a \neq b,
                 i = 0, \ldots, q-1 \}. 
\end{eqnarray*}

We will call the $\df$ \emph{uniform} if all the $Q_i$ are of the same
size, and we will say it is a \emph{perfect}\footnote{The term
  \emph{perfect} is used in \cite{handbook} to refer to a specific
  type of difference family where half the differences cover half the
  ground set.  Our usage is found in \cite{TonchevDSS} in relation to
  self-synchronising codes.}  internal (or external) $\df$ if
$|\intl(d)|$ (or $|\extl(d)|$) is a constant for all $d \in \abelian\setminus\{0\}$.
We will call the $\df$ a \emph{partition type} $\df$ if $\{Q_0, \ldots,
Q_{q-1}\}$ is a partition of $\abelian$.

\begin{remark}
As mentioned before and as will be pointed out in Section
\ref{sec:ddf}, there is by no means a consensus on the terms used to
describe a $\df$.  Here we point out the disparity between our terms and
those of \cite{handbook}, and in Section \ref{sec:ddf} we will point
out the differences as they arise.
In particular, the definition of difference family in \cite{handbook}
  stipulates that the subsets $Q_i$ are all of the same size, but does
  not insist that they are disjoint.  We have defined a $\df$ to consist
  of disjoint subsets (of varying sizes) because we want to be able to
  define \emph{external} differences.  Using the term \emph{uniform}
  to describe the subsets $Q_i$ being of the same size is consistent
  with terminology used in design theory.
\end{remark}

\begin{example} \label{eg:diffset}
A $(v,k, \lambda)$-difference set $Q_0$ over $\Z_v$ is a perfect
internal $\df(v;k)$ with $|\intl(d)| = \lambda = k(k-1)/(v-1)$.  If we
let $Q_1=\Z_v \setminus Q_0$ then $(\Z_v; Q_0, Q_1)$ is an internal
$\df(v; k, v-k)$ with $|\intl_0(d)| = \lambda$ and
$|\intl_1(d)|=v-2k+\lambda$ for all $d \in \Z_v^*$.  In fact, $(\Z_v;
Q_0, Q_1)$ has $|\intl(d)|=v-2k+2\lambda$ and $|\extl(d)| = v(v-1) -
(v-2k+2\lambda)$ and is a perfect internal and external $\df$.

For example, the $(7, 3, 1)$ difference set $Q_0=\{0,1,3\} \subseteq
\Z_7$.  We have $|\intl(d)| = 1$.  Let $Q_1=\Z_7 \setminus Q_0
=\{2,4,5,6 \}$.  Then $(\Z_7;Q_0, Q_1)$ has $|\intl_0(d)|=1$, $|\intl_1(d)|=2$, $|\intl(d)|=3$, $|\extl_0(d)| =
|\extl_1(d)|=2$ and $|\extl(d)|=4$ for all $d \in \Z_7^*$.
\end{example}

It is not hard to see that a perfect partition type internal $\df$ is
also a perfect partition type external $\df$ and vice versa.  However,
this is not generally true for $\df$s that are not partition type:

\begin{example} \label{eg:notperfect}
Let $\abelian = \Z_{25}$, $Q_0 = \{1,2,3,4,6,15\}$, $Q_1=\{5,9,10,14,17,24\}$.
This is a perfect external $\df(25; 6,6)$ with $|\extl(d)|=3$ for all $d \in \Z_{25}^*$ given in \cite{selfsyn}.  However, it is not a perfect internal $\df$: 
$$|\intl(d)| =\left\{
\begin{array}{ccl}
4 & \mbox{ for } & d=1, 24, \\
2 &  \mbox{ for } & d= 7,9,10,15,16,18, \\ 
1 &  \mbox{ for } & d= 6,8,17,19, \\
3 & \mbox {for} & \mbox{all other } d. 
\end{array} \right.
$$
\end{example}

For many codes and sequences \cite{FHS,AMD,selfsyn,auth,OOC}, desirable
properties can be expressed in terms of (some) external or internal
differences of $\df$s.  We give a brief survey of these
applications and the properties required of the $\df$s in
the next section.

\section{Disjoint difference families in applications}
\label{sec:ddf}

This is not intended to be a comprehensive survey of where disjoint
difference families arise in applications, nor of each application.
We want to show that these objects arise in many areas of
communications and information security research and that a study of
their various properties may be useful in making advances in these
fields.

\subsection{Frequency hopping (FH) sequences}
\label{sub:FHS}

Let $F=\{f_0, \ldots, f_{q-1}\}$ be a set of frequencies used in a
frequency hopping multiple access communication system
\cite{spreadspectrum}.  A frequency hopping (FH) sequence $X$ of
length $v$ over $F$ is simply $X=(x_0, x_1, \ldots, x_{v-1})$, $x_i
\in F$, specifying that frequency $x_i$ should be used at time $i$.
If two FH sequences use the same frequency at the same time (a
collision), the messages sent at that time may be corrupted.
Collisions are given by Hamming correlations: if a single sequence
together with all its cyclic shifts are used then we are interested in
its auto-correlation values (the number of positions in which each
cyclic shift agrees with the original sequence).  If two or more
sequences are used then it is also necessary to consider the
cross-correlation between pairs of sequences (the number of positions
in which cyclic shifts of one sequence agree with the other sequence
in the pair).

A single FH sequence $X$ may be viewed in a combinatorial way: Define
$Q_i$, $i =0, \ldots, q-1$, as subsets of $\Z_v$, with $j \in Q_i$ if
$x_j = i$.  Hence each $Q_i$ corresponds to a frequency $f_i$, and the
elements of $Q_i$ are the positions in $X$ where $f_i$ is used.  (For
example, the frequency hopping sequence $X=(0,0,1,0,1,1,1)$ over
$F=\{0, 1\}$ gives the $\df$ of Example \ref{eg:diffset}.)  In
\cite{FHS} it was shown that an FH sequence $(x_0, x_1,
\ldots, x_{v-1})$ with out-of-phase auto-correlation value of at most
$\lambda$ exists if and only if $(\Z_v; Q_0, \ldots, Q_{q-1})$ is a
partition type $\df(v; k_0, \ldots, k_{q-1})$ with $\intl(d)$
satisfying

$$|\intl(d)| \le \lambda \mbox{ for all } d \in \Z_v^*.$$
In \cite{FHS} $(\Z_v; Q_0, \ldots, Q_{q-1})$ is called a 
partition type difference packing.  

The aim in FH sequence design is to minimise collisions: we would like
$\lambda$ to be small.  Lempel and Greenberger
\cite{lempelgreenberger} proved a lower bound for $\lambda$, and in
\cite{PengFan} bounds relating the size of sets of frequency hopping
sequences with their Hamming auto- and cross-correlation values were
given.  Lempel and Greenberger \cite{lempelgreenberger} constructed
optimal sequences using transformations of m-sequences (more details
in Section \ref{sub:LG}).  In \cite{FHS} Fuji-Hara \emph{et al.} also
provided many examples of optimal sequences using designs with certain
types of automorphisms.  Other constructions of FHS include using
cyclotomy \cite{cyclotomy,generalcyclotomy}, random walks on expander
graphs \cite{randomwalk}, and error-correcting codes
\cite{FHSerrorcodes,FHScycliccodes}.  A survey of sequence design from
the viewpoint of codes can also be found in \cite{SarwateCDMA}.  Later
in this paper we will show that one of the constructions in \cite{FHS}
by Fuji-Hara \emph{et al.} gave the same sequences as those
constructed by Lempel and Greenberger in \cite{lempelgreenberger}. It
would be interesting to see how the other constructions relate to each other.

Note that in this correspondence to a difference family, the set of
frequency hopping sequence is the rotational closure (Definition
\ref{def:rotational}, Section \ref{sec:equiv}) of one single frequency
hopping sequence.  Collections of $\df$s were used to model more
general sets of sequences in
\cite{balancednested2,balancednested1,mixeddiff},  referred to as
balanced nested difference packings.

It is also to be noted that most of the published work considered
either pairwise interference between two sequences (described above
as Hamming correlation) or adversarial interference (jamming)
\cite{randomwalk,jammingfeasibility,jammingsurvey}, which may not
reflect the reality of the application where more than two sequences
may be in use.  To this end Nyirenda \emph{et al.} \cite{Mwawi}
modelled frequency hopping sequences as cover-free codes and
considered additional properties required to resist jamming.

\subsection{Self-synchronising codes}
\label{sub:selfsyn}

Self-synchronising codes are also called comma-free codes and have the
property that no codeword appears as a substring of two concatenated
codewords.  This allows for synchronisation without external help.
Codes achieving self-synchronisation in the presence of up to
$\lfloor \frac{\lambda-1}{2} \rfloor$ errors can be constructed from a
$\df(v; k_0, \ldots, k_{q-1})$ $(\Z_v; Q_0, \ldots, Q_{q-1})$ with
$|\extl(d)| \ge \lambda$.  In \cite{selfsyn}, this
$\df$ was called a {\em difference system of sets} of index
$\lambda$ over $\Z_v$.  The sets $Q_0, \ldots, Q_{q-1}$ give the markers for
self-synchronisation and are a redundancy, hence we would like
$k=\sum_{i=0}^{q-1} k_i$ to be small.  Other optimisation problems
include reducing the rate $k/v$, reducing $\lambda$, and reducing the
number $q$ of subsets.

An early paper by Golomb \emph{et al.} \cite{golomb-comma-free} took the
combinatorial approach to the subject of self-synchronising codes, and
\cite{comma-free} gave a survey of results, constructions and open
problems of self-synchronising codes. More recent work on
self-synchronising codes can be found in \cite{self-reflective} which
gave some variants on the definitions, and in \cite{nonoverlapping}
in the guise of non-overlapping codes, giving constructions and
bounds. Further constructions can be found in \cite{selfsyn}, including
constructions from the partitioning of cyclic difference sets and
partitioning of hyperplanes in projective geometry, as well as
iterative constructions using external and internal $\df$s.

\subsection{Splitting A-codes and secret sharing schemes with cheater detection}
\label{sub:acodes}

In authentication codes (A-codes), a transmitter and a receiver share
an encoding rule $e$, chosen according to some specified probability
distribution.  To authenticate a source state $s$, the transmitter
encodes $s$ using $e$ and sends the resulting message $m=e(s)$ to the
receiver.  The receiver receives a message $m'$ and accepts it if it is a valid encoding of some source, {\it i.e.} when
$m'=e(s^\prime)$ for some source $s^\prime$. In a splitting A-code, the message is computed with an
input of randomness so that a source state is not uniquely mapped to a
message. An adversary (who does not know which encoding rule is being used) may send their own message $m$ to the receiver in the hope that it will be accepted as valid.  This is known as an  {\em impersonation attack}, and succeeds if $m$ is a valid encoding of some source $s$.  Also of concern are {\em substitution attacks}, in which an adversary who has seen an encoding $m$ of a source $s$ replaces it with a new value $m^\prime$.  This attack succeeds if $m^{\prime}$ is a valid encoding of some source $s^\prime \neq s$.  We refer to \cite{auth} for further background.  It was shown
in \cite{auth} that optimal splitting A-codes can be constructed from
a perfect uniform external $\df(v; k_0=k, \ldots, k_{q-1}=k)$ with
$|\extl(d)| = 1$.  This gives an A-code with $q$ source states, $v$
encoding rules, $v$ messages, and each source state can be mapped to
$k$ valid messages.  This type of $\df$ was called an external difference
family (EDF) in \cite{auth}.  The probability of an
adversary successfully impersonating the transmitter is given by
$kq/v$ and the probability of successfully substituting a message
being transmitted is given by $1/kq$ (which also happens to equal
$k(q-1)/(v-1)$ in this particular context).  These are parameters to
be minimised.

An extensive list of A-code references prior to 1998 is given in \cite{acodebib}.  More recent work on splitting authentication codes includes \cite{blundo1999fallacious,ding2003three,ding2004three,ge2005combinatorial,Huber09,Huber10,ftcitHuber10,huber2010combinatorial,Huber10arxiv,kurosawa2001combinatorial,liang2011new,liang2012new,obana2001bounds,pei2006authentication,wang2006new,wang2010further}

A {\em secret sharing scheme} is a means of distributing some information, known as {\em shares}, to a set of {\em players} so that authorised subsets of players are able to combine their shares to reconstruct a unique secret, whereas the shares belonging to unauthorised subsets reveal no information about the secret.  If some of the players are dishonest, however, then they may {\em cheat} by submitting false values that are not their true shares and thereby causing an incorrect value to be obtained during secret reconstruction. Such attacks were first discussed by Tompa and Woll in \cite{TompaW88}.  Various types of difference family have been used in constructing schemes which allow such cheating to be detected with high probability.  In \cite{OgataKS06}, difference sets were used to construct schemes that were optimal with respect to certain bounds on the sizes of shares.  In \cite{auth}, EDFs were used in a similar manner to construct optimal schemes.  Other schemes that permit detection of cheaters include those proposed in \cite{ArakiO07,CabelloPS02,HoshinoO15,ObanaA06,Obana11,ObanaT14}.  Many of these constructions can be interpreted as involving particular types of difference family; this observation has led to the definition of the concept of {\em algebraic manipulation detection codes} \cite{CramerDFPW08} (see Section~\ref{sub:wAMD}).

\subsection{Weak algebraic manipulation detection (AMD) codes}
\label{sub:wAMD}

An AMD code is a tool that can be combined with a cryptographic system that provides some
form of secrecy in order to incorporate extra robustness against an adversary who can actively change values in the system.  The notion was proposed in \cite{CramerDFPW08} as an abstraction of techniques used in the construction of robust secret sharing schemes.  In the basic setting for a weak AMD code, a {\em source} is chosen uniformly from a finite set $S$ of sources with $|S|=k$.  It is then encoded using a (possibly randomised) encoding map $E\colon S \rightarrow \abelian$ where $\abelian$ is an abelian group of order $v\geq k$. We require the sets of possible encodings of different sources to be disjoint, so that $E(s)$ uniquely determines $s$.  An adversary is able to manipulate this encoded value by adding a group element $d\in \abelian$ of its choosing.  (We suppose the adversary knows the details of the encoding function, but does not know what source has been chosen, nor the specific value of any randomness used in the encoding.)  After this manipulation, an attempt is made to decode the resulting value.  If the altered value $E(s)+d$ is a valid encoding $E(s^\prime)$ of some source $s^\prime$ then it is decoded to $s^\prime$.  Otherwise, decoding fails and the symbol $\perp$ is returned; this represents the situation where the adversary's manipulation has been detected.  The adversary is deemed to have succeeded if $E(s)+d$ is decoded to $s^\prime \neq s$, that is if they have caused the stored value to be decoded to a source other than the one that was initially stored.  

A set of sources $S$ with $|S|=k$, abelian group $\abelian$ with $|\abelian|=v$ and encoding rule $E$ constitute a {\em weak $(k,
v, \epsilon)$-AMD (algebraic manipulation detection) code} if for any choice of $d\in \abelian$ the adversary's success probability is at most $\epsilon$.  (The probability is taken over the uniform choice of source, and over the randomness used in the encoding.)

In \cite{CramerDFPW08}, it was shown that a weak $(k, v, \epsilon)$-AMD code with deterministic encoding is
equivalent to a $\df(v; k)$ with
$$|\intl(d)| \le \lambda, \; \lambda \le \epsilon k\; \mbox{ for all } d
\in \abelian.$$

In \cite{CramerDFPW08} these were called $(v, k, \lambda)$-bounded difference
sets.  It is easy to see that these are generalisations of difference
sets, allowing general abelian groups and with an upper bound for the
number of differences.  

Weak AMD codes were introduced in \cite{CramerDFPW08}, with further detail on constructions, bounds and applications provided in the full version of the paper \cite{CramerDFPW08eprint}.  
Bounds on the adversary's success probability in a weak AMD code were
given in \cite{AMD} and several families with good asymptotic
properties were constructed using vector spaces.  Additional bounds
were given in \cite{AMD15}, and constructions and characterisations were
given relating weak AMD codes that are optimal with respect to these
bounds to a variety of types of external $\df$.  It is desirable to
minimise the tag length ($\log v - \log k$, the number of redundant
bits) as well as $\epsilon$.

\subsection{Stronger forms of algebraic manipulation detection (AMD) code}
\label{sub:AMD}

{\em Strong AMD codes} were defined in \cite{CramerDFPW08}; these are able to limit the success probability of an adversary even when the adversary knows which source has been encoded.  Specifically, for
every source $s \in S$ and every element $d \in \abelian$, the
probability that $(E(s)+d)$ is decoded to a value $s^\prime \not \in \{ s, \perp \}$ is at most
$\epsilon$.  (Here the probability is taken over the randomness in the encoding rule $E$.  Unlike the case of a weak AMD code, a strong AMD code cannot use a deterministic encoding rule.)

  Write $Q_i = \{ g \in \abelian \; : \; D(g) = s_i \}$ for
each $s_i \in S$, $i = 0, \ldots, k-1$, and $|Q_i|=k_i$.  In the case
where the encoding $E(s_i)$ is uniformly distributed over $Q_i$ for
every $s_i$, we have that $(\abelian; Q_0, \ldots, Q_{k-1})$ forms a $\df(v; k_0,
\ldots, k_{k-1})$ with $|\extl_i(d)| \le \lambda_i =\epsilon k_i$ and
$|\extl(d)| \le \lambda = \sum_{i=0}^{k-1} \lambda_i$. 

 Constructions
from vector spaces and caps in projective space were given in
\cite{AMD}. 
Additional bounds and characterisations were given in \cite{AMD15}.  A construction based on a polynomial over a finite field was given in \cite{CramerDFPW08} and applied to the construction of robust secret sharing schemes, and robust fuzzy extractors.  This construction has since been used for a range of applications, including the construction of anonymous message transmission schemes \cite{broadbent2007information}, non-malleable codes \cite{DziembowskiPW10}, strongly decodeable stochastic codes \cite{GuruswamiSmith}, secure communication in the presence of a byzantine relay \cite{he2009secure,he2013strong}, and codes for the adversarial wiretap channel \cite{wang2014efficient}.  New constructions, including an asymptotically optimal randomised construction were given in \cite{CramerPX15}.

AMD codes that resist adversaries who learn some limited information about the source were constructed and analysed in \cite{AhmadiS13}, and their application to tampering detection over wiretap channels was discussed.

AMD codes secure in a stronger model in which an adversary succeeds even when producing a new encoding of the original source have been used in the design of secure cryptographic devices and related applications \cite{ge2013reliable,ge2013secure,karpovsky2014design,luo2013secure,luo2014hardware,wang2012new,WangK11}.

\subsection{Optical orthogonal codes (OOCs)}
\label{sub:OOC}

Optical orthogonal codes (OOCs) are sequences arising from
applications in code-division multiple access in fibre optic channels.
OOC with low auto- and cross-correlation values allow users to
transmit information efficiently in an asynchronous environment.  A
$(v,w,\lambda_a, \lambda_c)$-OOC of size $q$ is a family $\{ X_0,
\ldots, X_{q-1}\}$ of $q$ $(0,1)$-sequences of length $v$, weight $w$,
such that auto-correlation values are at most $\lambda_a$ and
cross-correlation values are at most $\lambda_c$.  For each sequence
$X_i$, let $Q_i$ be the set of integers modulo $v$ denoting the
positions of the non-zero bits.  Then $(\Z_v; Q_0, \ldots, Q_{q-1})$
is a uniform $\df(v; k_0=w \ldots, k_{q-1}=w)$ with
\begin{eqnarray*}
|\intl_i(d)| & \le & \lambda_a, \\
|\extl_{i,j}(d)| & \le & \lambda_c, \mbox{ for all } d \in \Z_v^*.
\end{eqnarray*}

Background and motivation to the study of OOC were given in \cite{OOC}, which also 
included constructions from designs, algebraic codes and projective geometry.  In \cite{bitan} constant weight cyclically permutable codes, which are also uniform $\df$s, were used to construct OOC, and a recursive construction was given.
In \cite{sensing} OOC were used to construct compressed sensing matrix and
a relationship between OOC and modular Golomb rulers (\cite{handbook})
was given - a $(v,k)$
modular Golomb ruler is a set of $k$ integers $\{d_0, \ldots,
d_{k-1}\}$ such that all the differences are distinct and non-zero
modulo $v$ - in fact, a $\df(v; k)$ with $|I(d)| \le 1$ for all $d
\neq 0$.

A generalisation to two-dimensional OOC with a combinatorial approach
can be found in \cite{2D-OOC,combinatorialOOC}.
Combinatorial and recursive constructions as well as bounds can be
found in \cite{boundOOC}, and \cite{variableOOC} allowed variable
weight OOC and used various types of difference families and designs
to construct such OOCs.

\subsection{Other applications}
\label{sub:other}

The list of applications discussed in this section is by no means 
exhaustive, and $\df$s arise in a variety of other areas of combinatorics
and coding theory.  For example, in \cite{CompleteDDF}, complete sets
of disjoint difference families (in fact, partition type perfect
uniform $\df$s where the subsets are grouped) were used in constructing
{\em 1-factorisations of complete graphs} and in constructing {\em cyclically
resolvable cyclic Steiner systems}. In \cite{quasicyclic}, {\em high-rate
quasi-cyclic codes} were constructed using perfect internal uniform $\df$, and a
generalisation to families of sets of non-negative integers with
specific internal differences was given. {\em $\mathbb{Z}$-cyclic whist tournaments} correspond to perfect internal $\df$s over $\mathbb{Z}_v$ \cite{whist}.  In addition, various types of sequences and arrays with specified correlation properties have been proposed for a wide range of applications \cite{golomb2005signal,sequencecorrelation}.  Many of these can be studied in terms of a relationship with appropriate forms of $\df$s  \cite{pott1999difference}.

\section{A geometrical look at a perfect partition type disjoint difference 
family}
\label{sec:geometric}

In \cite{FHS} a perfect partition type 
$\df(q^{n}-1; k_0=q-1,k_1=q, \ldots, k_{q^{n-1}-1}=q)$ over $\Z_{q^n-1}$ 
was constructed from line orbits of a cyclic perspectivity $\tau$ in
the $n$-dimensional projective space $PG(n,q)$ over $\gf{q}$.  In
\cite{lempelgreenberger} another construction with the same parameters
was given.  In the next section we will show a correspondence between
the two constructions.  Before that we will describe in greater detail
the the construction of \cite[Section III]{FHS}.

An $n$-dimensional projective space $PG(n,q)$ over the finite field of
order $q$ admits a cyclic group of perspectivities $\group$ of order
$q^n-1$ that fixes a hyperplane $\hyperpl_{\infty}$ and a point
$\infty \notin \hyperpl_{\infty}$.  (We refer the reader to
\cite{Hirschfeld} for properties of projective spaces and their
automorphism groups.)  This group $\group$ acts transitively on the
points of $\hyperpl_{\infty}$ and regularly on the points of $PG(n,q)
\setminus (\hyperpl_{\infty} \cup \{\infty\})$.  We will call the
points (and spaces) not contained in $\hyperpl_{\infty}$ the affine
points (and spaces).

The point orbits of $\group$ are
$\{ \infty \}$,
$\hyperpl_{\infty}$, and
$PG(n,q) \setminus (\hyperpl_{\infty} \cup \{\infty\})$.
Dually, the hyperplane orbits are 
$\hyperpl_{\infty}$, the set of all hyperplanes through $\infty$, and 
the set of all hyperplanes of 
$PG(n,q) \setminus \hyperpl_{\infty}$ not containing $\infty$.
Line orbits under $\group$ are:

\begin{enumerate} [label=(\Alph*)] 
\item \label{thruinfty} One orbit of affine lines through $\infty$ -
  this orbit has length $\frac{q^n - 1}{q-1}$; and
\item \label{notthruinfty} $\frac{q^{n-1} - 1}{q-1}$ orbits of affine
  lines not through $\infty$ - each orbit has length $q^n-1$, and
  $\group$ acts regularly on each orbit; and
\item \label{infty} One orbit of lines contained in $\hyperpl_{\infty}$.
\end{enumerate}

A set of parallel (affine) lines through a point 
$P_{\infty} \in \hyperpl_{\infty}$ 
consists of one line $L_0$ from the orbit of type
\ref{thruinfty} and $q-1$ lines from each of the $(q^{n-1}-1)/(q-1)$
orbits of type \ref{notthruinfty}.  We will write this set of
$q^{n-1}$ lines ${\cal{P}} = \{L_0, L_1, \ldots, L_{q^{n-1}-1}\}$ as
follows (See Figure \ref{pic:parallel}):
\begin{itemize}
\item $L_0$, a line through $\infty$ and $P_{\infty} \in
\hyperpl_{\infty}$;
\item $\orb_i = \{ L_{(i-1)(q-1)+1}, L_{(i-1)(q-1)+2}, \ldots,
  L_{(i-1)(q-1)+(q-1)} \}$, $i = 1, \ldots, \frac{q^{n-1}-1}{q-1},$
  each $\orb_i$ belonging to a different orbit under $\group$.
\end{itemize}

\begin{figure} \caption{The parallel class $\packing$.}
\label{pic:parallel}
\begin{center}
\begin{tikzpicture}[scale=1]
%\draw[step=1cm,gray,very thin] (0,0) grid (10,10);
%The hyperplane 
\draw (5,8) ellipse (4cm and 1cm);

%The lines in a parallel class
\draw  (1,2) -- (5,8);
\draw  (3,2) -- (5,8);
\draw  (4,2) -- (5,8);
\draw  (5,2) -- (5,8);
\draw  (7,2) -- (5,8);
\draw  (8,2) -- (5,8);
\draw  (9,2) -- (5,8);

%Somne points + label
\fill[black] (2,3.5) circle (0.1cm);
\fill[black] (5,8) circle (0.1cm);
%Labels
\node at (1.6,3.5) {$\infty$};
\node at (5,8.4) {$P_{\infty}$};
\node at (5.8,3.5) {$\ldots$};
\node at (8,9) {$\hyperpl_{\infty}$};

\node at (1,1.5) {$L_0$};
\node at (4,1.5) {$ \underbrace{\hspace{2cm}} $};
\node at (4,1) {$\orb_1$};
\node at (8,1.5) {$ \underbrace{\hspace{2cm}} $};
\node at (8,1) {$\orb_{\frac{q^{n-1}-1}{q-1}}$};

\end{tikzpicture}
\end{center}
\end{figure}

We consider the two types of $d \in \Z_{q^n-1}^*$ depending on the
action of $\tau^d$ on $L_0$:

\begin{enumerate}[label=(\Roman*)] 
\item \label{fixL0} There are $q-2$ values of $\tau^d$, $d \in \Z_{q^n-1}^*$,
  fixing the line $L_0$ (and the points $P_{\infty}$ and $\infty$) and
  permuting the points of $L_0$.  These $\tau^d$ permute but do not
  fix the lines within each $\orb_i$.  Hence we have, for these $d \in
  \Z_{q^n-1}^*$,
\begin{eqnarray}
\label{eqnarray:1}
L_0^{\tau^d} \cap L_0 = L_0 \mbox{ and } L_i^{\tau^d} \cap L_i = \{P_\infty\}.
\end{eqnarray}

\item \label{notfixL0}
The remaining $(q^n-1) - (q-2)$ values of $\tau^d$ map lines in $\cal{P}$ to
affine lines not in $\cal{P}$.  Hence we have
\begin{eqnarray}
\label{eqnarray:2}
L_0^{\tau^d} \cap L_0 = \{\infty\} \mbox{ and } 
|L_i^{\tau^d} \cap L_i| = 0 \mbox{ or } 1.
\end{eqnarray}

Without loss of generality consider $L_1 \in \orb_1$.  Suppose
$|L_1^{\tau^d} \cap L_1| = 1$, say $L_1^{\tau^d} \cap L_1 = \{P\}$.
Let $L_k \in \orb_1$ be another line in the same orbit as $L_1$, so
there is a $d_k$ such that $L_1^{\tau^{d_k}} = L_k$.  It is not hard to see
that $\{P^{\tau^{d_k}}\} =  L_k^{\tau^d} \cap L_k$, since
\begin{eqnarray*}
P \in L_1 &\Rightarrow& P^{\tau^{d_k}} \in L_1^{\tau^{d_k}} = L_k, \\
P \in L_1^{\tau^d} &\Rightarrow& P^{\tau^{d_k}} \in (L_1^{\tau^d})^{\tau^{d_k}}
= (L_1^{\tau^{d_k}})^{\tau^{d}} = L_k^{\tau^{d}}.
\end{eqnarray*}
Hence for any orbit $\orb_i$,
if $|L_j^{\tau^d} \cap L_j| = 1$ for some $L_j \in \orb_i$ 
then $|L_k^{\tau^d} \cap L_k|=1$ for all $L_k \in \orb_i$

Now, suppose again that $|L_1^{\tau^d} \cap L_1| = 1$.  Let $P_1$ be
the point on $L_1$ such that $P_1^{\tau^d} \in L_1^{\tau^d} \cap L_1$.
Consider $L_j \in \orb_k$, $k\neq 1$.  Suppose $|L_j^{\tau^d} \cap
L_j| = 1$.  Let $P_2$ be the point on $L_j$ such that $P_2^{\tau^d}
\in L_j^{\tau^d} \cap L_j$.  (See Figure \ref{pic:difforb}.)  Since
$\group$ is transitive on affine points (excluding $\infty$), there is
a $d_j$ such that $P_1^{\tau^{d_j}} = P_2$.  Then
$$ (P_1^{\tau^d})^{\tau^{d_j}} = (P_1^{\tau^{d_j}})^{\tau^{d}} =
P_2^{\tau^d}.$$ 
This means that $\tau^{d_j}$ maps $P_1$ to $P_2$ and
$P_1^{\tau^d}$ to $P_2^{\tau^d}$ and hence maps the line $L_1$ to
$L_j$.  But this is a contradiction since $L_1$ and $L_j$ belong to
different orbits under $\group$.  Hence if $|L_j^{\tau^d} \cap L_j| = 1$
for any $L_j$ in some orbit $\orb_i$ then $|L_k^{\tau^d} \cap L_k| = 0$ for all 
$L_k$ in all other orbits.

\begin{figure} \caption{$L_1$, $L_j$ in different orbits}
\label{pic:difforb}
\begin{center}
\begin{tikzpicture}[scale=1]
%\draw[step=1cm,gray,very thin] (0,0) grid (10,10);
%The hyperplane 
\draw (5,8) ellipse (4cm and 1cm);
\node at (8,9) {$\hyperpl_{\infty}$};

%The 4 lines + labels
\draw  (3,3) -- (3,8); %L_1
\node at (3,2.5) {$L_1$};
\draw  (1,5) -- (7,8); %L_1^\tau^d
\node at (1,4.5) {$L_1^{\tau^d}$};
\draw  (8,3) -- (3,8); %L_j
\node at (8.5,2.5) {$L_j$};
\draw  (5,2) -- (7,8); %L_j^\tau^d
\node at (4.8,1.5) {$L_j^{\tau^d}$};

%Somne points + label
\fill[black] (3,8) circle (0.1cm); %P_\infty
\node at (3,8.4) {$P_{\infty}$};
\fill[black] (7,8) circle (0.1cm); %P_\infty^\tau^d
\node at (7,8.4) {$P_{\infty}^{\tau^d}$};
\fill[black] (3,4) circle (0.1cm); %P_1
\node at (3.5,4) {$P_1$};
\fill[black] (3,6) circle (0.1cm); %P_1^\tau^d
\node at (3.5,6) {$P_1^{\tau^d}$};
\fill[black] (7,4) circle (0.1cm); %P_2
\node at (7.5,4) {$P_2$};
\fill[black] (6,5) circle (0.1cm); %P_2^\tau^d
\node at (6.5,5) {$P_2^{\tau^d}$};

%group action
\draw [->] (3,4) to [out=60,in=300] (3.1,5.9);
\node at (3.6,5) {$\tau^d$};
\draw [->] (7,4) to [out=180,in=270] (6,4.9);
\node at (6,4.2) {$\tau^d$};
\draw [dashed,->] (3,4) to [out=330,in=210] (6.8,3.9);
\node at (4.3,3.2) {$\tau^{d_j}$};
\end{tikzpicture}
\end{center}
\end{figure}
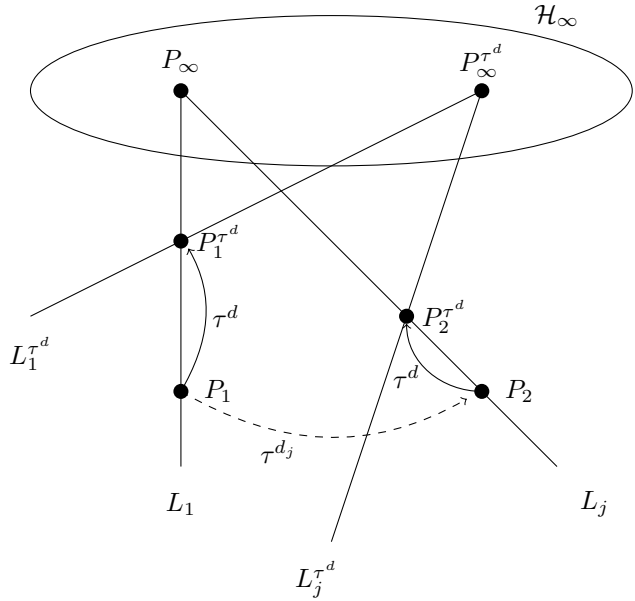

It is also clear that for any $L_i$ in any orbit, there is a $d$ such
that $|L_i^{\tau^d} \cap L_i| = 1$, because $\group$ is transitive on
affine points (excluding $\infty$).  Indeed, $\group$ acts regularly
on these points, so that for any pair of points $(P,Q)$ on $L_i$ there
is a unique $d$ such that $P^{\tau^d} = Q$.  There are $q(q-1)$ pairs
of points and so there are $q(q-1)$ such values of $d$.  
These $q(q-1)$ values of $d$ for each $\orb_i$ in $\cal{P}$, 
together with the $q-2$ values of $d$ where $\tau^d$ that fixes $L_0$, 
account for all of $\Z_{q^n-1}^*$.
\end{enumerate}

Now, the points of $PG(n,q) \setminus (\hyperpl_{\infty} \cup \{\infty\})$
can be represented as $\Z_{q^n-1}$ as follows: pick an arbitrary point
$P_0$ to be designated $0$.  The point $P_0^{\tau^i}$ corresponds to
$i \in \Z_{q^n-1}$.  The action of $\tau^d$ on any point $P$ is thus
represented as $P+d$.  Affine lines are therefore $q$-subsets
of $\Z_{q^n-1}$.  
Let $Q_0 \subseteq \Z_{q^n-1}$ contain the points of $L_0 \setminus
\{\infty\}$, and let $Q_i$ contain the points of $L_i$.  It follows
from the intersection properties of the lines (properties (\ref{eqnarray:1}),
(\ref{eqnarray:2})) that $\{Q_0, \ldots, Q_{q^{n-1} -1} \}$ forms a
perfect partition type $\df(q^{n}-1; q-1,q, \ldots, q)$ over
$\Z_{q^n-1}$, with $|\intl(d)| = q-1$ for all $d \in \Z_{q^n-1}^*$.

\subsection{A perfect external $\df$}

Given that a partition type perfect internal $\df$ over
$\Z_v$ with $|\intl(d)| = \lambda$ must be a perfect external
$\df$ with $|\extl(d)| = v-\lambda$, the intersection
properties $|L_i^{\tau^d} \cap L_j|$, $i \neq j$ can be deduced as
follows for the two different types \ref{fixL0}, \ref{notfixL0} of $d$:

\begin{enumerate} [label = (\Roman*)]
\item For the $q-2$ values of $\tau^d$ of type \ref{fixL0} fixing $L_0$,  we have:
\begin{enumerate} [label = (\alph*)]
\item  $L_0$ is fixed, so $|L_0^{\tau^d} \cap L_i| = 0$ for all $L_i \neq L_0$.
\item  If $L_i$ and $L_j$ are in different orbits then  
$|L_i^{\tau^d} \cap L_j| = 0$ (since $\tau^d$ fixes $\orb_i$).
\item If $L_i$ and $L_j$ are in the same orbit, then since $\tau^d$
  acts regularly on an orbit of type \ref{notthruinfty}, there is a
  unique $d$ that maps $L_i$ to $L_j$, so $|L_i^{\tau^d} \cap L_j| =
  q$, and for all other $L_k$ in the same orbit, $|L_i^{\tau^d} \cap
  L_k| = 0$.  This applies to each orbit, so that for each of the
  $q-2$ values of $d$, there are $((q^{n-1}-1)/(q-1)) \times (q-1) =
  q^{n-1}-1$ cases where $|L_i^{\tau^d} \cap L_j| = q$.
\end{enumerate}
\item For the $(q^n-1)-(q-2)$ values of $\tau^d$ of type \ref{notfixL0} not 
fixing $L_0$, we have:
\begin{enumerate}[label = (\alph*)]
\item Pick any point $P \in L_0 \setminus \{\infty\}$, $P^{\tau^d} \in
  L_i$ for some $L_i \neq L_0$, so $|L_0^{\tau^d} \cap L_i| = 1$ for
  some $L_i$. There are $q-1$ points on $L_0 \setminus \{\infty\}$, so
  there are $q-1$ lines $L_i$ such that $|L_0^{\tau^d} \cap L_i| = 1$.
\item Consider $L_i \neq L_0$.  Take any point $P \in L_i$.  We have
  $P^{\tau^d} \in L_j$ for some $L_j$, so $|L_i^{\tau^d} \cap L_j| =
  1$.  This applies for all $L_i$, so that for any of the
  $(q^n-1)-(q-2)$ values of $d$, there are $(q^{n-1}-1)q$ cases of
  $|L_i^{\tau^d} \cap L_j| = 1$, $q-1$ of which are when $L_j = L_i$.
\end{enumerate}
\end{enumerate}

Defining the sets $Q_0, \ldots, Q_{q^-1}$ as before, we see that 
$\{Q_0, \ldots, Q_{q^-1}\}$ forms a perfect partition type $\df$
with $\extl(d) = q(q^{n-1}-1)$.

\section{A correspondence between two difference families}
\label{sec:corr}

In \cite{FHS}, Fuji-Hara {\it et al.}\ constructed the perfect
partition type $\df(q^{n}-1; q-1,q, \ldots, q)$ over
$\Z_{q^n-1}$ with $|\intl(d)| = q-1$ described in Section
\ref{sec:geometric}.  Using parallel $t$-dimensional subspaces (we
described the case when $t=1$), perfect partition type 
$\df(q^n-1; q^t-1, q^t, \ldots, q^t)$ with $|\intl(d)| =
q^t-1$ can also be constructed.

This construction gives $\df$ with the same parameters as those
constructed using m-sequences in \cite{lempelgreenberger}, though
\cite{lempelgreenberger} restricted their constructions to the case
when $q$ is a prime.  It was asked in \cite{FHS} whether these are
``essentially the same'' constructions.  In this section we show a
correspondence between these two constructions, and in Section
\ref{sec:equiv} we discuss what ``essentially the same'' might mean.
This correspondence also shows that the restriction to $q$ prime in
\cite{lempelgreenberger} is unnecessary.  (Indeed it was pointed out
in \cite{SarwateCDMA} that the assumption that the field must be prime
is not necessary.)

\subsection{The Lempel-Greenberger m-sequence construction}
\label{sub:LG}

We refer the reader to \cite{lidl} for more details on linear
recurring sequences.  Here we sketch an introduction.  Let
$(s_t)=s_0s_1s_2\ldots$ be a sequence of elements in $\gf{q}$, $q$ a
prime power, satisfying the $n^{\rm th}$ order linear recurrence relation
\begin{eqnarray*}
s_{t+n} = c_{n-1}s_{t+n-1} + c_{n-2}s_{t+n-2}+ \cdots + c_0s_t,
\; c_i \in \gf{q}, \; c_{n-1} \neq 0.
\end{eqnarray*}
Then $(s_t)$
is called an ($n^{\rm th}$ order) linearly recurring sequence in $\gf{q}$.
Such a sequence can be generated using
a \emph{linear feedback shift register (LFSR)}.  An LFSR is a device 
with $n$ \emph{stages}, which we denote by $S_0, \ldots ,S_{n-1}$.
Each stage is capable of storing one element of $\gf{q}$.
The contents $s_{t+i}$ of all the registers $S_i$ ($0 \le i \le n-1$)
at a particular time $t$ are known as
the \emph{state} of the LFSR at time $t$.
We will write it either as  
$s(t,n)=s_{t}s_{t+1} \ldots s_{t+n-1}$ or as a vector 
$\s_t = (s_t, s_{t+1}, \ldots, s_{t+n-1})$.  The state 
$\s_0 =(s_0,s_1, \ldots, s_{n-1})$ is the \emph{initial state}.

At each clock cycle, an output from the LFSR is extracted and the 
LFSR is updated as described below.
\begin{itemize}
\item
The content $s_t$ of the stage $S_0$ is output and forms 
part of the \emph{output sequence}.
\item
For all other stages,
the content $s_{t+i}$ of stage $S_i$ is moved to stage $S_{i-1}$
($1 \le i \le n-1$).
\item
The new content $s_{t+n}$ of stage $S_{n-1}$ is the 
value of the \emph{feedback function}
\begin{eqnarray*}
f(s_t,s_{t+1},\ldots,s_{t+n-1}) = 
c_0s_t + c_1s_{t+1} + \cdots c_{n-1}s_{t+n-1},\; c_i \in \gf{q}.
\end{eqnarray*}
The new state is thus $\s_{t+1}=(s_{t+1}, s_{t+2}, \ldots, s_{t+n})$.
The constants $c_0, c_1, \ldots, c_{n-1}$ are known as the
\emph{feedback coefficients} or \emph{taps}.
\end{itemize}

A diagrammatic representation of an LFSR is given in Figure~\ref{fig:LFSR}.

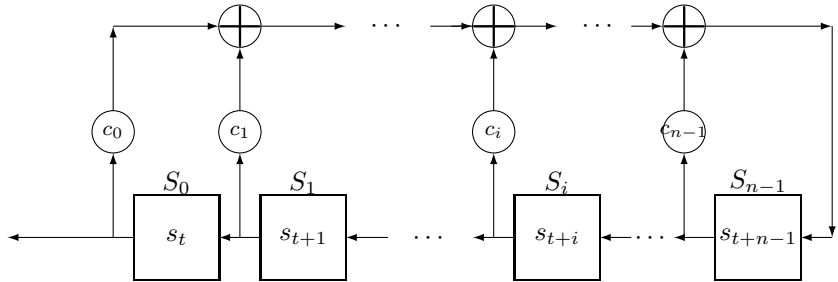
\begin{figure}[ht]
\begin{center}
\unitlength=0.8pt
\begin{picture}(440.00,164.00)(0.00,0.00)
\put(50.00,144.00){\vector(1,0){50.00}}
\put(50.00,104.00){\vector(0,1){40.00}}
\put(50.00,94.00){\circle{20}}
\put(50.00,44.00){\vector(0,1){40.00}}
\put(50.00,94.00){\makebox(0.00,0.00){\footnotesize $c_0$}}
\put(110.00,144.00){\circle{20}}
\put(110.00,144.00){\makebox(0.00,0.00){\Huge +}}
\put(120.00,144.00){\vector(1,0){40.00}}
\put(110.00,104.00){\vector(0,1){30.00}}
\put(110.00,94.00){\circle{20}}
\put(110.00,44.00){\vector(0,1){40.00}}
\put(110.00,94.00){\makebox(0.00,0.00){\footnotesize$c_1$}}
%\put(170.00,144.00){\circle{20}}
%\put(170.00,144.00){\makebox(0.00,0.00){\Huge +}}
\put(320.00,144.00){\circle{20}}
\put(320.00,144.00){\makebox(0.00,0.00){\Huge +}}
\put(330.00,144.00){\vector(1,0){60.00}}
\put(320.00,104.00){\vector(0,1){30.00}}
\put(320.00,94.00){\circle{20}}
\put(320.00,44.00){\vector(0,1){40.00}}
\put(320.00,94.00){\makebox(0.00,0.00){\footnotesize $c_{n-1}$}}
\put(180.00,144.00){\makebox(0.00,0.00){$\cdots$}}
\put(200.00,144.00){\vector(1,0){20.00}}
%\put(180.00,144.00){\vector(1,0){50.00}}
%\put(170.00,104.00){\vector(0,1){30.00}}
%\put(170.00,94.00){\circle{20}}
%\put(170.00,44.00){\vector(0,1){40.00}}
%\put(170.00,94.00){\makebox(0.00,0.00){\footnotesize $c_2$}}
\put(230.00,144.00){\circle{20}}
\put(230.00,144.00){\makebox(0.00,0.00){\Huge +}}
\put(210.00,144.00){\vector(1,0){50.00}}
\put(295.00,144.00){\vector(1,0){15.00}}
\put(280.00,144.00){\makebox(0.00,0.00){$\cdots$}}
\put(230.00,104.00){\vector(0,1){30.00}}
\put(230.00,94.00){\circle{20}}
\put(230.00,44.00){\vector(0,1){40.00}}
\put(230.00,94.00){\makebox(0.00,0.00){\footnotesize $c_i$}}
\put(390.00,144.00){\vector(0,-1){100.00}}
\put(390.00,44.00){\vector(-1,0){15.00}}
\put(335.00,44.00){\vector(-1,0){20.00}}
\put(305.00,44.00){\makebox(0.00,0.00){$\cdots$}}
\put(295.00,44.00){\vector(-1,0){15.00}}
\put(200.00,44.00){\makebox(0.00,0.00){$\cdots$}}
\put(240.00,44.00){\vector(-1,0){20.00}}
%\put(240.00,44.00){\vector(-1,0){20.00}}
\put(180.00,44.00){\vector(-1,0){20.00}}
\put(120.00,44.00){\vector(-1,0){20.00}}
\put(60.00,44.00){\vector(-1,0){60.00}}
\put(335.00,24.00){\framebox(40.00,40.00){$s_{t+n-1}$}}
\put(240.00,24.00){\framebox(40.00,40.00){$s_{t+i}$}}
%\put(180.00,24.00){\framebox(40.00,40.00){$s_{t+2}$}}
\put(120.00,24.00){\framebox(40.00,40.00){$s_{t+1}$}}
\put(60.00,24.00){\framebox(40.00,40.00){$s_t$}}
\put(335.00, 50.00){\makebox(40.00,40.00){$S_{n-1}$}}
\put(240.00,50.00){\makebox(40.00,40.00){$S_i$}}
%\put(180.00,50.00){\makebox(40.00,40.00){$S_2$}}
\put(120.00,50.00){\makebox(40.00,40.00){$S_1$}}
\put(60.00,50.00){\makebox(40.00,40.00){$S_0$}}
\end{picture}
\end{center}
\caption{Linear Feedback Shift Register}
\label{fig:LFSR}
\end{figure}

The \emph{characteristic polynomial} associated with the LFSR (and the
linear recurrence relation) is
\begin{eqnarray*} 
f(x) = x^n - c_{n-1} x^{n-1} - c_{n-2} x^{n-2} - \cdots - c_0.  
\end{eqnarray*}

The state at time $t+1$ is also given by $\s_{t+1}= \s_t C$, where $C$ is the
\emph{state update matrix} given by

\begin{eqnarray*}
C = \left( \begin{array}{ccccc}
0 & 0 & \ldots & 0 & c_0 \\
1 & 0 & \ldots & 0 & c_1 \\
0 & 1 & \ldots & 0 & c_2 \\
\vdots & \vdots & \ddots & \vdots & \vdots \\
0 & 0 & \ldots & 1 & c_{n-1}
\end{array} \right) .
\end{eqnarray*}

A sequence $(s_t)$ generated by an $n$-stage LFSR is periodic and has
maximum period $q^n-1$.  A sequence that has maximum period is
referred to as an \emph{m-sequence}.  An LFSR generates an m-sequence
if and only if its characteristic polynomial is primitive.  An
m-sequence contains all possible non-zero states of length $n$, hence
we may use, without loss of generality, the {\em impulse response sequence}
(the sequence generated using initial state $(0\cdots01)$).

Let $S = (s_t) = s_0s_1s_2 \ldots$ be an m-sequence over a prime field
$\gf{p}$ generated by an $n$-stage LFSR with a primitive
characteristic polynomial $f(x)$.  Let $s(t,k) = s_t s_{t+1} \ldots
s_{t+k-1}$ be a subsequence of length $k$ starting from $s_t$.

The $\sigma_k$-transformations, $1 \le k \le n-1$ introduced in
\cite{lempelgreenberger} are described as follows:
\begin{eqnarray*}
\sigma_k : s(t,k) = s_t s_{t+1} \ldots s_{t+k-1} \to 
\sum_{i=0}^{k-1} s_{t+i}p^i \in \Z_{p^k} = \{ 0, 1, \ldots, p^k - 1\}.
\end{eqnarray*}

We write the $\sigma_k$-transform of $S$ as $U = (u_t)$, 
$u_t = \sigma_k( s(t,k) )$, which is a sequence over $\Z_{p^k}$.

In \cite[Theorem 1]{lempelgreenberger} it is shown that the sequence
$U$ forms a frequency hopping sequence with out-of-phase
auto-correlation value of $p^{n-k}-1$, and hence a partition type
perfect $\df$ with $|\intl(d)| = p^{n-k}-1$
(Section \ref{sub:FHS}).  We see in the next section that this
corresponds to the geometric construction of \cite{FHS} described in
Section \ref{sec:geometric}.

\subsection{A geometric view of the Lempel-Greenberger m-sequence construction.}
\label{sub:LG-FMM}

We refer the reader to \cite{Hirschfeld} for details about coordinates in finite
projective spaces over $\gf{q}$.  Here we only sketch what is
necessary to describe the m-sequence construction of Section
\ref{sub:LG} from the projective geometry point of view.  

Let $PG(n,q)$ be an $n$-dimensional projective space over $\gf{q}$.
Then we may write
\begin{eqnarray*} 
PG(n,q) = \{ (x_0, x_1, \ldots, x_n) \; | \; x_i \in \gf{q} 
\mbox{ not all zero}\},  
\end{eqnarray*}
with the proviso that $\rho(x_0, x_1, \ldots, x_n)$ as $\rho$ ranges over $\gf{q}
\setminus \{0\}$ all refer to the same point.  
Dually a hyperplane of $PG(n,q)$ is written as $[a_0, a_1, \ldots,
  a_n]$, $a_i \in \gf{q}$ not all zero, and contains the points $(x_0,
x_1, \ldots, x_n)$ satisfying the equation
\begin{eqnarray*}
 a_0x_0 + a_1x_1 + \cdots + a_nx_n = 0.
\end{eqnarray*}
Clearly $\rho[a_0, a_1, \ldots, a_n]$ as $\rho$ ranges over $\gf{q} \setminus
\{0\}$ refers to the same hyperplane.  A $k$-dimensional subspace is
specified by either the points contained in it, or the equations of
the $n-k$ hyperplanes containing it.

Now, let $S = (s_t) = s_0s_1s_2 \ldots$ be an m-sequence over
$\gf{p}$, $p$ prime, generated by an $n$-stage LFSR with a primitive
characteristic polynomial $f(x)$ and state update matrix $C$, as
described in the previous section.  For $t=0, \ldots, p^n-2$, let $P_t
= (s_t, s_{t+1}, \ldots, s_{t+n-1}, 1)$. Then the set $\orb=\{P_t \; |
\; t=0, \ldots p^n-2 \}$ are the points of $PG(n,p) \setminus (
\hyperpl_{\infty} \cup \{ \infty \})$ where $\hyperpl_{\infty}$ is the
hyperplane $x_{n}=0$ and $\infty$ is the point $(0, \ldots, 0, 1)$.

Let $\tau$ be the projectivity defined by 
\begin{eqnarray*}
A = \left( \begin{array}{cccc}
  &  &  &  0 \\
  & C & &  \vdots \\
  &  &  &  0 \\
0 & \ldots & 0 & 1
\end{array} \right).
\end{eqnarray*}

Then $\tau$ fixes $\hyperpl_{\infty}$ and $\infty$, acts regularly on
$\orb=\{P_t \; | \; t=0, \ldots, p^n-2 \}$, and maps $P_t$ to $P_{t+1}$.
Now we consider what a $\sigma_k$-transformation means in $PG(n,p)$.

Firstly we consider $\sigma_{n-1}$.  This takes the first $n-1$ coordinates 
of the point $P_t = (s_t, s_{t+1}, \ldots, s_{t+n-1}, 1)$ and maps them to 
$\sum_{i=0}^{n-2} s_{t+i} p^i \in \Z_{p^{n-1}}$.  There are $p^{n-1}$ 
distinct $z_i \in \Z_{p^{n-1}}$ and for each $z_i \neq 0$, there are $p$
points $Z_i=\{P_{t_0}, \ldots, P_{t_{p-1}} \} = 
\{(s_t, s_{t+1}, \ldots, s_{t+n-2}, \alpha, 1) \; | \; \alpha \in \gf{p}\}$ 
which are mapped to $z_i$ by 
$\sigma_{n-1}$.  For $z_i=0$ there are $p-1$ corresponding points in $Z_0$ 
since the all-zero state does not occur in an m-sequence.

It is not hard to see that the sets $Z_0 \cup \{\infty\}$, $Z_1,
\ldots Z_{p^{n-1}-1}$ form the set of parallel (affine) lines through
the point $(0,\ldots, 0,1,0) \in \hyperpl_{\infty}$, since $Z_i$ is the 
set $\{(s_t, \ldots, s_{t+n-2}, \alpha , 1)\;|\; \alpha \in \gf{p} \}$ for
some $(n-1)$-tuple $(s_t, \ldots, s_{t+n-2})$ and this forms a line
with $(0,\ldots, 0,1,0) \in \hyperpl_{\infty}$ (the line defined by
the $n-1$ hyperplanes $x_0 - s_t x_n=0$, $x_1 - s_{t+1} x_n = 0$,
$\ldots$, $x_{n-2} - s_{t+n-2} x_n = 0$).  This is precisely the
construction given by \cite{FHS} described in Section
\ref{sec:geometric}.  For each $Z_i =\{P_{t_0}, \ldots, P_{t_{p-1}}
\}$, $i=1, \ldots, p^{n-1}-1$, let $D_i = \{t_0, \ldots, t_{p-1} \}$,
and for $Z_0=\{P_{t_0}, \ldots, P_{t_{p-2}}\}$, let $D_0 = \{t_0,
\ldots, t_{p-2} \}$.  Then the sets $D_i$ form a partition type perfect
internal $\df(p^n; p-1, p, \ldots, p)$ over
$\Z_{p^n}$ with $|\intl(d)| = p-1$ for all $d \in \Z_{p^n}^*$.

Similarly, for $\sigma_k$, $1 \le k \le n-1$, the set of points
\begin{eqnarray*}
Z_i =\{ (s_t, \ldots, s_{t+n-k-1}, \alpha_1, \ldots, \alpha_k, 1) \; | \;
\alpha_1, \ldots, \alpha_k \in \gf{p}\}
\end{eqnarray*}
corresponding to each $z_i \in \Z_{p^k}$ form an $(n-k)$-dimensional 
subspace and the set of $Z_i$
forms a parallel class.  These are the constructions of \cite[Lemma 3.1,
3.2]{FHS}.

\begin{example}
Let $S = (s_t)$ be an m-sequence over $\gf{3}$ satisfying the linear recurrence relation $x_{t+3} = 2x_t + x_{t+2}$.  The state update matrix is therefore
\begin{eqnarray*}
 C = \left( \begin{array}{ccc}
  0  & 0 &  2 \\
  1  & 0 &  0  \\
  0  & 1 &  1 
\end{array} \right).
\end{eqnarray*}

The impulse response sequence is $S=(00111021121010022201221202)$,
and the $\sigma_3$-, $\sigma_2$- and $\sigma_1$-transformations give

\begin{center}
\begin{tabular}{cccccc} \hline 
$P_t$ & $s(t,3)$ & $\sigma_3(s(t,3))$ & $s(t,2)$ & $\sigma_2(s(t,2))$ & $s(t,1)=\sigma_1(s(t,1))$ \\ \hline
$P_0$ &   001  &     9    &     00  &     0 & 0 \\
$P_1$ &   011  &     12   &     01  &     3 & 0 \\
$P_2$ &   111   &    13   &      11  &     4 & 1 \\
$P_3$ &   110  &     4     &    11  &     4 & 1 \\
$P_4$ &   102  &     19   &     10  &     1 & 1 \\
$P_5$ &   021  &     15   &     02  &     6 & 0 \\
$P_6$ &   211  &     14   &     21  &     5 & 2 \\
$P_7$ &   112  &     22   &     11  &     4 & 1 \\
$P_8$ &   121  &     16   &     12  &     7 & 1 \\
$P_9$ &   210  &     5    &     21  &     5 & 2 \\
$P_{10}$ &  101  &     10    &    10   &    1 & 1 \\
$P_{11}$ &  010  &     3    &     01  &     3 & 0 \\
$P_{12}$ &  100   &    1    &     10   &    1 & 1 \\
$P_{13}$ &  002  &     18   &     00   &    0 & 0 \\
$P_{14}$ &  022  &     24   &     02   &    6 & 0 \\
$P_{15}$ &  222  &     26   &     22   &    8 & 2 \\
$P_{16}$ &  220  &     8    &     22    &   8 & 2 \\
$P_{17}$ &  201  &     11   &     20    &   2 & 2 \\
$P_{18}$ &  012  &     21   &     01    &   3 & 0 \\
$P_{19}$ &  122   &    25    &    12    &   7 & 1 \\
$P_{20}$ &  221   &    17   &     22    &   8 & 2 \\
$P_{21}$ &  212   &    23   &     21    &   5 & 2 \\
$P_{22}$ &  120   &    7    &     12    &   7 & 1 \\
$P_{23}$ &  202  &     20    &    20    &   2 & 2 \\
$P_{24}$ &  020  &     6    &     02    &   6 & 0 \\
$P_{25}$ &  200  &     2    &     20    &   2 & 2
\end{tabular}
\end{center}

Writing this in $PG(3,3)$, $P_t = (s_t, s_{t+1}, s_{t+2}, 1)$, and 
$\hyperpl_{\infty}$ is the hyperplane $x_3=0$, and $\infty$ is the point
$(0,0,0,1)$.  The projectivity $\tau$ maps $P_t$ to $P_{t+1}$, where
$\tau$ is represented by the matrix $A$,
\begin{eqnarray*}
A = \left( \begin{array}{cccc}
  0  & 0 &  2 & 0 \\
  1  & 0 &  0 & 0 \\
  0  & 1 &  1 & 0 \\
  0  & 0 & 0 & 1
\end{array} \right).
\end{eqnarray*}

The $\sigma_2$ transformation maps 3 points to every $z_i \in \Z_9^*$.
These form the affine lines of $PG(3,3)$ through the point
$(0,0,1,0)$.  For example, the points $P_{1}, P_{11}, P_{18}$ lie on
the line defined by $x_0=0$, $x_1-x_3=0$.  The set $\{1, 11, 18\}$
would be one of the subsets of the difference family.  This gives
$Q_0=\{0,13\}$, $Q_1=\{1, 11, 18\}$, $Q_2=\{5,14,24\}$,
$Q_3=\{4,10,12\}$, $Q_4=\{2,3,7\}$, $Q_5 = \{8, 19, 22\}$, $Q_6 =
\{17, 23, 25\}$, $Q_7=\{6,9,21\}$, $Q_8=\{15, 16, 20\}$.

The $\sigma_1$ transformation maps 9 points to every $z_i \in \Z_3^*$.
These form the affine planes of $PG(3,3)$ through the point
$(0,0,1,0)$. For example, the points $P_2$, $P_3$, $P_4$, $P_7$, 
$P_8$, $P_{10}$,
$P_{12}$, $P_{19}$, $P_{22}$ lie on the plane $x_0-x_3=0$.  The sets
\begin{eqnarray*}
Q_0&=&\{0,1,5,11,13,14,18,24\},\\ 
Q_1&=& \{2,3,4,7,8,10,12,19,22\},\\
Q_2&=&\{6,9,15,16,17,20,21,23,25\}
\end{eqnarray*}
 form a difference family over $\Z_3$.

\end{example}

\subsection{The other way round?}
\label{sub:otherway}

We see that the m-sequence constructions of \cite{lempelgreenberger}
gives the projective geometry constructions of \cite{FHS}.  Here we
consider how the constructions of \cite{FHS} relate to m-sequences.

In $PG(n,q)$ we may choose any $n+2$ points (every set of $n+1$ of
which are independent) as the \emph{simplex of reference} (there is an
automorphism that maps any set of such $n+2$ points to any other set).
Hence we may choose the hyperplane $x_n=0$ (denoted $\hyperpl_{\infty}$)
and the point $(0,0,\ldots, 0,1)$ (denoted $\infty$).

Now, consider a projectivity $\tau$ represented by an $(n+1) \times (n+1)$
matrix $A$ that fixes $\hyperpl_{\infty}$ and $\infty$.  It must take
the form
\begin{eqnarray*}
A = \left( \begin{array}{cccc}
  &  &  &  0 \\
  & C & &  \vdots \\
  &  &  &  0 \\
0 & \ldots & 0 & 1
\end{array} \right),
\end{eqnarray*}
and we see that 
\begin{eqnarray*}
A^i = \left( \begin{array}{cccc}
  &  &  &  0 \\
  & C^i & &  \vdots \\
  &  &  &  0 \\
0 & \ldots & 0 & 1
\end{array} \right). 
\end{eqnarray*}
So the order of $A$ is given by the order of $C$. Let the
characteristic polynomial of $C$ be $f(x)$.  The order of $A$ is hence
the order of $f(x)$.

Consider the action of $\group$ on the points of $PG(n,q) \setminus
(\hyperpl_{\infty} \cup \infty)$. For $\group$ to act transitively on these
points $A$ must have order $q^n-1$, which means that $f(x)$ must be
primitive.  If we use this $f(x)$ as the characteristic polynomial for
an LFSR we generate an m-sequence, as in Section \ref{sub:LG}. For
prime fields, this is precisely the construction of
\cite{lempelgreenberger}.

Projectivities in the same conjugacy classes have matrices that are
similar and therefore have the same characteristic polynomial.  There
are $\frac{\phi(q^n-1)}{n}$ primitive polynomials of degree $n$ over
$\gf{q}$ and this gives the number of conjugacy classes of
projectivities fixing $\hyperpl_{\infty}$ and $\infty$ and acting
transitively on the points of $PG(n,q) \setminus (\hyperpl_{\infty}
\cup \infty)$.

For a particular $\group$ with characteristic polynomial $f(x)$ and
difference family $\{Q_0, \ldots, Q_{q^{n-1}-1} \}$, there are $q^n-1$
choices for the point $P_0$ to be designated $0$ in the construction
described in Section \ref{sec:geometric}.  Each choice gives $Q_i + d$
for each $Q_i$, $i = 1, \ldots, q^{n-1}-1$, $d \in \Z_{q^n-1}^*$.
This corresponds to the $q^n-1$ shifts of the m-sequence generated by
the LFSR with characteristic polynomial $f(x)$.  The choice of
parallel class (the point $P_{\infty} \in \hyperpl_{\infty}$) gives
the difference family $\{Q_i + d \; : \; d \in \Z_{q^n-1},i = 1,
\ldots, q^{n-1}-1 \}$.  (There are $q-1$ values of $d$ such that
$\{Q_i + d\} = \{Q_i\}$.)  This corresponds to a permutation of symbols
and a shift of the m-sequence.  If the set of shifts of an m-sequence
is considered as a cyclic code over $\gf{q}$ then this gives
equivalent codes (more on this in Section \ref{sec:equiv}).  The
group $\group$ has $\phi(q^n-1)$ generators, and each of the generators
$\tau^i$, $(i, q^n-1)=1$ corresponds to a multiplier $w$ such that
$\{w Q_i \; : \; i = 1, \ldots, q^{n-1}-1 \} = \{Q_i \; : \; i = 1,
\ldots, q^{n-1}-1 \}$.
  
We have described this correspondence in terms of the lines of
$PG(n,q)$ but this also applies to the correspondence between higher
dimensional subspaces and the $\sigma_k$-transformations.

\begin{example}
In $PG(3,3)$, the group of perspectivities generated by $\tau$,
represented by the matrix
$$ A = \left( \begin{array}{cccc} 0 & 1 & 0 & 0 \\ 1 & 0 & 1 & 0
  \\ 0 & 1 & 1 & 0 \\ 0 & 0 & 0 & 1 \end{array} \right),$$ fixes the
plane $x_3=0$ and fixes the point $\infty=(0,0,0,1)$.  An affine point
$(x,y,z,1)$ is mapped to the point $(y,x+z,y+z,1)$ and a plane
$[a,b,c,d]$ is mapped to the plane $[a+b-c,a,-a+c,d]$.  Taking the
point $(1,0,0,1)$ as $0$, we have the affine lines through
$P_{\infty}=(1,0,0,0)$ in $x_3=0$ as
$$ \begin{array}{lclcl} 
Q_0 = \{ 0, 13\}, & \hspace{1cm} & Q_1 = \{1,19,4\}, & \hspace{1cm} &
Q_2 = \{2,22,23\}, \\
Q_3 = \{ 3,5,12\}, & \hspace{1cm} & Q_4 = \{6,14,17\}, & \hspace{1cm} &
Q_5 = \{7,11,21\}, \\
Q_6 = \{8,24,20\}, & \hspace{1cm} & Q_7 = \{9,10,15\}, & \hspace{1cm} &
Q_8 = \{16,18,25\}.
\end{array}$$
If we consider the action of $\tau^5$, we have
$$ \begin{array}{lcl} 
Q'_0 = \{ 0, 13\} = Q_0 \times 7, 
& \hspace{1cm} & Q'_1 = \{6,9,21\} = Q_3 \times 7, \\
Q'_2 = \{16,20,15\} = Q_4 \times 7, 
& \hspace{1cm} & Q'_3 = \{ 11,1,18\} = Q_7 \times 7, \\  
Q'_4 = \{22,8,19\} = Q_8 \times 7, 
& \hspace{1cm} & Q'_5 = \{17,23,25\} = Q_5 \times 7, \\
Q'_6 = \{12,10,4\} = Q_6 \times 7, 
& \hspace{1cm} &  Q'_7 = \{2,3,7\} = Q_1 \times 7, \\ 
Q'_8 = \{24,14,5\} = Q_2 \times 7. & &
\end{array}$$
If we choose a different parallel class, say, $P'_{\infty}=(0,0,1,0)$, we
will instead have 
$$ \begin{array}{lclcl} 
Q''_0 = \{ 10, 23\}, & \hspace{1cm} & Q''_1 = \{1,24,16\}, & \hspace{1cm} &
Q''_2 = \{2,0,9\}, \\
Q''_3 = \{ 3,14,11\}, & \hspace{1cm} & Q''_4 = \{4,8,18\}, & \hspace{1cm} &
Q''_5 = \{5,17,21\}, \\
Q''_6 = \{6,7,12\}, & \hspace{1cm} & Q''_7 = \{13,15,22\}, & \hspace{1cm} &
Q''_8 = \{19,20,25\},
\end{array}$$
and $\{Q'_0, \ldots, Q'_8\} = \{Q_0+10, \ldots, Q_8+10\}$.

The characteristic polynomial of $A$ is $f(x) = x^3 - x^2 - 2x - 2$.
Using $f(x)$ as the characteristic polynomial of an LFSR we have the
update matrix $C$ as
$$ C = \left( \begin{array}{ccc} 0 & 0 & 2 \\ 1 & 0 & 2 \\ 0 & 1 &
  1 \end{array} \right).$$ 
Using the process described in
Section \ref{sub:LG-FMM}, we obtain (with $(0,0,1,1)$ as $0$) the
difference family $\{Q_i -1\;: \; i=0, \ldots 8\}$. 
\end{example}

It is clear from this correspondence that the m-sequence constructions
of \cite{lempelgreenberger} also works over a non-prime field.  The
$\sigma_k$ transform is essentially assigning a unique symbol to each
$k$-tuple from the initial m-sequence.

\section{Equivalence of FH sequences}
\label{sec:equiv}

In \cite{FHS},  Fuji-Hara \emph{et al.} stated ``Often we are
interested in properties of FH sequences, such as auto-correlation,
randomness and generating method, which remain unchanged when passing
from one FH sequence to another that is essentially the same.
Providing an exact definition for this concept and enumerating how
many non `essentially the same' FH sequences are also interesting
problems deserving of attention.''  Here we discuss the notion of 
equivalence of FH sequences.

Firstly we adopt the notation of \cite{Mwawi} for frequency hopping schemes:
An $(n,M,q)$-frequency hopping scheme (FHS) $\fhs$ is a set of $M$ words of
length $n$ over an alphabet of size $q$.  Each word is an FH sequence.

Elements of the symmetric group $S_n$ can act on $\fhs$ by permuting
the coordinate positions of each word in $\fhs$.  Let $\rho_n$
denote the permutation $\begin{pmatrix}1&2&\cdots&n\end{pmatrix}\in
  S_n$.  We say that an element of $S_n$ is a {\em rotation} if it
  belongs to $\langle \rho_n \rangle$, the subgroup generated by
  $\rho_n$.

\begin{example}
Consider the $(7,1,2)$-FHS ${\fhs}$ consisting of the single word 
$(0,0,0,1,0,1,1)$.  We have $(0,0,0,1,0,1,1)^{\rho_7}=(1,0,0,0,1,0,1)$.
\end{example}

\begin{definition} \label{def:rotational}
Let $Q$ be a finite alphabet.  Given a set $S\subseteq Q^n$ we define the 
{\em rotational closure} of $S$ to be the set 
\begin{eqnarray*}
{\rot{S}}=\{\mathbf{w}^\sigma\mid \mathbf{w}\in S,\ 
\sigma\in\langle\rho_n\rangle\}.
\end{eqnarray*}
If $\rot{S}=S$ then we say that $S$ is {\em rotationally closed}.
\end{definition}

\begin{example}
Consider again the binary $(7,1,2)$-FHS $\fhs$ consisting of the single word
$(0,0,0,1,0,1,1)$.  Its rotational closure is the orbit of the word
$(0,0,0,1,0,1,1)$ under the action by the subgroup $\langle
\rho_7\rangle$:
\begin{eqnarray*}
\rot{\fhs}=\{&(0,0,0,1,0,1,1),\\ &(1,0,0,0,1,0,1),\\ &(1,1,0,0,0,1,0),\\&(0,1,1,0,0,0,1),\\&(1,0,1,1,0,0,0),\\&(0,1,0,1,1,0,0),\\&(0,0,1,0,1,1,0)\}.
\end{eqnarray*}
\end{example}

If $\fhs$ is a FHS then $\rot{\cal F}$ is precisely the set of
sequences available to users for selecting frequencies.  An important
property of a FHS is the Hamming correlation properties of the sequences in
$\fhs$.

Let $\fhs$ be an $(n,M,q)$-FHS and let $\mathbf{x}=(x_0, \ldots, x_{n-1})$, 
$\mathbf{y}=(y_0, \ldots, y_{n-1}) \in \fhs$. The Hamming correlation 
$H_{\mathbf{x},\mathbf{y}}(t)$ at relative time delay $t$, $0 \le t < n$,  
between $\mathbf{x}$
 and $\mathbf{y}$ is
$$ H_{\mathbf{x},\mathbf{y}}(t) = \sum_{i=0}^{n-1} h(x_i, y_{i+t}), $$
where
$$h(x,y) = \left\{ \begin{array}{ll}
1 & \mbox{if } x = y,\\ 
0 & \mbox{if } x \neq y.
\end{array} \right.$$

Note that the operations on indices are performed modulo $n$. If
$\mathbf{x} = \mathbf{y}$ then $H_{\mathbf{x}}(t) =
H_{\mathbf{x},\mathbf{x}}(t)$ is the Hamming auto-correlation.  The
maximum out-of-phase Hamming auto-correlation of $\mathbf{x}$ is
\[ H(\mathbf{x}) = \max_{1 \le t<n } \{H_{\mathbf{x}}(t)\} \]
and the maximum Hamming cross-correlation between any two distinct 
FH sequences $\mathbf{x}$, $\mathbf{y}$ is
\[ H(\mathbf{x},\mathbf{y}) = \max_{0 \le t <n}
\{ H_{\mathbf{x},\mathbf{y}}(t)\}. \]
We define the maximum Hamming correlation of an  $(n,M,q)$-FHS $\fhs$ as 

\[M(\fhs) = \max_{\mathbf{x}, \mathbf{y} \in \fhs} 
\{H(\mathbf{x}),H(\mathbf{y}),H(\mathbf{x},\mathbf{y})\}. \]

\begin{theorem}
Let $\mathbf{w}\in Q^n$.  The maximum out-of-phase Hamming auto-correlation $H(\w)$ of 
$\mathbf{w}$ is equal to $n-d$, where $d$ is the minimum (Hamming) 
distance of $\rot{\mathbf{w}}$.
\end{theorem}

\begin{theorem}
Let $\fhs$ be an $(n,M,q)$-FHS.  The minimum distance of 
$\rot{\fhs}$ is equal to $n-M(\fhs)$.
\end{theorem}

The proofs of these theorems are trivial, but the theorems
suggest that taking the rotational closure of a frequency
hopping sequence allows us to work with the standard notion of Hamming
distance in place of the Hamming correlation.

\begin{theorem}
Let $\mathbf{w}\in Q^n$.  If $|\rot{\mathbf{w}}|<n$ then $H(\mathbf{w})=n$.
\end{theorem}

\begin{proof}
We observe that $|\rot{\mathbf{w}}|$ is the size of the orbit of
$\mathbf{w}$ under the action of the subgroup $\langle\rho_n\rangle$,
which has order $n$.  By the orbit-stabiliser theorem, if
$|\rot{\mathbf{w}}|<n$ then the stabiliser of $\mathbf{w}$ is
nontrivial.  That is, there is some (non-identity) rotation that maps
$\mathbf{w}$ onto itself.  This implies that its maximum out-of-phase
Hamming auto-correlation is $n$.
\end{proof}

In other words, unless a given sequence of length $n$ has worst
possible Hamming auto-correlation, its rotational closure always has
size $n$.

The following lemma is also straightforward to prove:
\begin{lemma}
Let $\mathbf{w}\in Q^n$.  If $|\rot{\mathbf{w}}|<n$ then  for $i=0,1,\dotsc,n-1$ we have $\rot{\mathbf{w}^{\rho_n^i}}=\rot{\mathbf{w}}$.
\end{lemma}

In coding theory, two codes are equivalent if one can be obtained from
the other by a combination of applying an arbitrary permutation to the
alphabet symbols in a particular coordinate position and/or permuting
the coordinate positions of the codewords.  These are transformations
that preserve the Hamming distance between any two codewords.  In the
case of frequency hopping sequences, it is the maximum Hamming correlation
that we wish to preserve.  This is a stronger condition, and hence the
set of transformations that are permitted in the definition of
equivalence will be smaller.  For example, we can no longer apply
different permutations to the alphabet in different coordinate
positions, as that can alter the out-of-phase Hamming correlations.  Because the
rotation of coordinate positions is inherent to the definition of
Hamming correlation, if we wish to permute the alphabet symbols then
we must apply the same permutation to the symbols in each coordinate
position.  Similarly, not all permutations of coordinates preserve the out-of-phase
Hamming auto-correlation of a sequence.
\begin{example}
Consider the sequence $(0,0,0,1,0,1,1)$.  Its maximum out-of-phase Hamming
auto-correlation is 3.  However, if we swap the first and last column
we obtain the sequence $(1,0,0,1,0,1,0)$, which has maximum out-of-phase Hamming
auto-correlation 5.
\end{example}
However, we can use the notion of rotational closure to determine an
appropriate set of column permutations that will preserve Hamming
correlation.  Recall that for a given word, its out-of-phase Hamming
auto-correlation is uniquely determined by the minimum distance of its
rotational closure.  Now, any permutation of coordinates preserves Hamming
distance, so if we can find a set of permutations that preserve the
property of being rotationally closed, then these will in turn
preserve the out-of-phase Hamming auto-correlation of individual sequences.

Suppose a word $\mathbf{w}$ of length $n$ has $H(\w)<n$.  Then its
rotational closure consists of the elements 
\begin{eqnarray*}\rot{\mathbf{w}}=\{\mathbf{w},
  \mathbf{w}^{\rho_n},\mathbf{w}^{\rho_n^2},\dotsc,\mathbf{w}^{\rho_n^{n-1}}\}.
\end{eqnarray*}
Applying a permutation $\gamma\in S_n$ to the coordinates of these words
gives the set
\begin{eqnarray*}
\left(\rot{\mathbf{w}}\right)^\gamma=\{\mathbf{w}^\gamma, \mathbf{w}^{\rho_n\gamma},\mathbf{w}^{\rho_n^2\gamma},\dotsc,\mathbf{w}^{\rho_n^{n-1}\gamma}\}.
\end{eqnarray*}
We wish to establish conditions on $\gamma$ that ensure that $\left(\rot{\mathbf{w}}\right)^\gamma$ is itself rotationally closed.
\begin{theorem}
Suppose $\mathbf{w}\in Q^n$ has out-of-phase Hamming auto-correlation
less than $n$.  Then $\left(\rot{\mathbf{w}}\right)^\gamma$ is
rotationally closed if and only if $\gamma\in
N_{S_n}(\langle\rho_n\rangle)$, that is $\gamma$ is an element of the
normaliser of $\langle \rho_n\rangle$ in $S_n$.
\end{theorem}

\begin{proof}
Suppose $\gamma\in N_{S_n}(\langle\rho_n\rangle)$ .  Then
$\gamma\langle\rho_n\rangle \gamma^{-1}= \langle \rho_n\rangle$.  This
implies that
\begin{eqnarray*}
\rot{\mathbf{w}}&=&\{\mathbf{w}^{\gamma\rho_n^i \gamma^{-1}} \mid 
i=0,1,2,\dotsc,n-1 \}, 
\end{eqnarray*}
and so
\begin{eqnarray*}
\left(\rot{\mathbf{w}}\right)^\gamma&=&
\{\mathbf{w}^{\gamma\rho_n^i } \mid i=0,1,2,\dotsc,n-1 \} \\
&=&\rot{\mathbf{w}^\gamma}.
\end{eqnarray*}
Conversely, if $\left(\rot{\mathbf{w}}\right)^\gamma$ is
rotationally closed, then
\begin{eqnarray*}
\left(\rot{\mathbf{w}}\right)^\gamma
&=&\{\mathbf{w}^\gamma, \mathbf{w}^{\rho_n\gamma},\mathbf{w}^{\rho_n^2\gamma},
\dotsc,\mathbf{w}^{\rho_n^{n-1}\gamma}\} \\
&=&\{\mathbf{w}', \mathbf{w}'^{\rho_n},\mathbf{w}'^{\rho_n^2},
\dotsc,\mathbf{w}'^{\rho_n^{n-1}}\}, 
\end{eqnarray*}
where $\mathbf{w}' = \mathbf{w}^{\rho_n^i\gamma}$ for some $i$.  So we have
\begin{eqnarray*}
\left(\rot{\mathbf{w}}\right)^\gamma
&=\{\mathbf{w}^{\rho_n^i\gamma}, \mathbf{w}^{\rho_n^i\gamma\rho_n},
\mathbf{w}^{\rho_n^i\gamma\rho_n^2},
\dotsc,\mathbf{w}^{\rho_n^i\gamma\rho_n^{n-1}}\}. 
\end{eqnarray*}
This means that $\mathbf{w}^\gamma = \mathbf{w}^{\rho_n^i\gamma\rho_n^j}$ for some $j$, and so $\mathbf{w}^{\gamma \rho_n^{-j} \gamma^{-1}} =  \mathbf{w}^{\rho_n^i}$. Clearly this applies to all $i$, $j$, and we have
$\gamma\in N_{S_n}(\langle\rho_n\rangle)$.
\end{proof}

\begin{example}

Consider the permutation 
$\gamma=\begin{pmatrix}2&5&3\end{pmatrix}
        \begin{pmatrix}4&6&7\end{pmatrix} \in S_7$.  
We have 
$\gamma^{-1}=\begin{pmatrix}2&3&5\end{pmatrix}
             \begin{pmatrix}4&7&6\end{pmatrix}$, and 
\begin{eqnarray*}
\gamma\rho_7\gamma^{-1}&=&\begin{pmatrix}2&5&3\end{pmatrix}\begin{pmatrix}4&6&7\end{pmatrix}\begin{pmatrix}1&2&3&4&5&6&7\end{pmatrix}\begin{pmatrix}2&3&5\end{pmatrix}\begin{pmatrix}4&7&6\end{pmatrix}\\
&=&\begin{pmatrix}1&3&5&7&2&4&6 \end{pmatrix}\\
&=&\rho_7^2.
\end{eqnarray*}
Since $\rho_7^2$ generates $\langle\rho_7\rangle$ this shows that $\gamma\in N_{S_7}(\langle\rho_7\rangle)$.

Now consider the word $(A,B,C,D,E,F,G)$.  The rows of the following matrix give its rotational closure:
\begin{eqnarray*}
\begin{bmatrix}
A&B&C&D&E&F&G\\
G&A&B&C&D&E&F\\
F&G&A&B&C&D&E\\
E&F&G&A&B&C&D\\
D&E&F&G&A&B&C\\
C&D&E&F&G&A&B\\
B&C&D&E&F&G&A\\
\end{bmatrix}.
\end{eqnarray*}
If we apply $\gamma$ to the columns of this matrix, we obtain
\begin{eqnarray*}
\begin{bmatrix}
A&C&E&G&B&D&F\\
G&B&D&F&A&C&E\\
F&A&C&E&G&B&D\\
E&G&B&D&F&A&C\\
D&F&A&C&E&G&B\\
C&E&G&B&D&F&A\\
B&D&F&A&C&E&G
\end{bmatrix},
\end{eqnarray*}
which is easily seen to be the rotational closure of any of its rows.

We now look at applying these ideas to the sequence $(0,0,0,1,0,1,1)$.
Permuting its coordinates with $\gamma$ in fact yields $(0,0,0,1,0,1,1)$,
which is trivially equivalent to the original sequence.  Less
trivially, $\begin{pmatrix}2&4&3&7&5&6\end{pmatrix}$ is another
  example of an element of the normaliser of $\langle \rho_7\rangle$,
  and applying this permutation to the coordinates yields the sequence
  $(0,1,1,0,1,0,0)$.  This is an example of an `equivalent' frequency
  hopping sequence that is not simply a rotation of the original
  sequence.
\end{example}

\begin{definition}\label{def:equiv}
We say that two $(n,M,q)$-FHSs are {\em equivalent} if one can be
obtained from the other by a combination of permuting the symbols of
the underlying alphabet and/or applying to the coordinates of its
sequences any permutation that is an element of
$N_{S_n}(\langle\rho_n\rangle)$.
\end{definition}
Equivalent FHSs have the same maximum Hamming correlation.

\subsection{Comparison with the notion of equivalence for $\df$s}

Two distinct difference families are said to be equivalent if there is
an isomorphism between the underlying groups that maps one $\df$ onto a
translation of the other.  In Section~\ref{sub:FHS} we discussed the
correspondence between a partition type $\df$ and an FHS.  In fact, we
will see that two partition type $\df$s over $\mathbb{Z}_n$ are
equivalent in this sense if and only if the corresponding FHSs are
equivalent in the sense of Definition~\ref{def:equiv}.  We begin by
noting that the automorphism group of $\mathbb{Z}_n$ is isomorphic to
$\mathbb{Z}_n^*$.  As in Section~\ref{sec:equiv} let $\rho_n\in S_n$
be the permutation $\begin{pmatrix}1&2&\cdots&n\end{pmatrix}$.  
Any element $\gamma\in
N_{S_n}(\langle\rho_n\rangle)$ induces a map $\phi_\gamma\colon
\mathbb{Z}_n\rightarrow \mathbb{Z}_n$ by sending $i\in \mathbb{Z}_n$
to the unique element $j\in \mathbb{Z}_n$ for which $\gamma^{-1}
\rho_n^i \gamma=\rho^j$.  The map $\phi_\gamma$ is a homomorphism,
since if $\phi_\gamma(i_1)=j_1$ and $\phi_\gamma(i_2)=j_2$ then
$\gamma^{-1}\rho_n^{i_1+i_2}\gamma=\gamma^{-1}\rho_n^{i_1}\gamma\gamma^{-1}\rho_n^{i_2}\gamma=\rho_n^{j_1}\rho_n^{j_2}=\rho_n^{j_1+j_2}$,
so $\phi_\gamma(i_1+i_2)=\phi_{\gamma}(i_1)+\phi_\gamma(i_2)$; in fact
it is an automorphism.  Every automorphism of $\langle \rho_n \rangle$
can be obtained in this fashion.
\begin{theorem} \label{thm:equiv}
Let $\cal F$ be a length $n$ FHS consisting of a single word, and let
$\cal D$ be the corresponding partition type $\df$ over
$\mathbb{Z}_n$.  Then the FHS obtained by applying a permutation
$\gamma\in N_{S_n}(\langle\rho_n\rangle)$ to the coordinate positions of
$\cal F$ corresponds to a $\df$ that is a translation of the $\df$
obtained from $\cal D$ by applying the automorphism $\phi_\gamma$ to
the elements of $\mathbb{Z}_n$.
\end{theorem}

\begin{proof}
  It is straightforward to verify that
  $\gamma^{-1}\rho_n\gamma$ is the cycle
  $\begin{pmatrix}1^\gamma&2^\gamma&\dotsc&n^\gamma\end{pmatrix}$.  For $\gamma\in
  N_{S_n}(\langle\rho_n\rangle)$ this is equal to $\rho_n^k$ for some
  $k$.  It follows that for $i=1,2,\dotsc,n-1$ we have
\begin{equation}\label{eq:gammas}
(i+1)^\gamma=i^\gamma+k.
\end{equation}
The correspondence between $\cal F$ and $\cal D$ is obtained by
associating positions in the sequence with elements of $\mathbb{Z}_n$.
For example, the FHS $\fhs = (1,1,2,3,2)$ corresponds to the $\df$
$(\mathbb{Z}_{5}; \{0,1\},\{2,4\},\{3\})$:
\begin{eqnarray*}
\begin{array}{rccccc}
\mathbb{Z}_5& 0&1&2&3&4\\
\hline
{\cal F} & 1&1&2&3&2
\end{array}.
\end{eqnarray*}

We observe that in this representation, the $i+1^{\rm th}$ element of the sequence $\cal F$ is in correspondence with the element $i\in \mathbb{Z}_n$.
If we apply $\gamma$ to the positions of $\cal F$, then the entry in the $j+1^{\rm th}$ position is mapped to the $i+1^{\rm th}$ position when $(j+1)^\gamma=i+1$.  Repeatedly applying the relation in \eqref{eq:gammas} tells us that in this case we have $i+1=1^\gamma+jk,$ so $i=(1^\gamma-1)+jk$.

If we apply $\phi_\gamma$ to $\mathbb{Z}_n$ then element $j\in \mathbb{Z}_n$ is replaced by element $i$ when $\gamma^{-1}\rho_n^j\gamma=\rho^i$.  But we have that 
\begin{eqnarray*}
\gamma^{-1}\rho_n^j\gamma&=&(\gamma^{-1}\rho_n\gamma)^j\\
&=&(\rho_n^k)^j\\
&=&\rho_n^{kj},
\end{eqnarray*}
so it must be the case that $i=kj$.  It follows that if we then translate this $\df$ by adding $1^\gamma-1$ to each element of $\mathbb{Z}$ we obtain the same overall transformation that was effected by applying $\gamma$ to $\cal F$.
\end{proof}

\begin{example}
For example, let $\gamma=\begin{pmatrix}1&5&3&4\end{pmatrix} \in N_{S_5}(\langle \rho_5\rangle)$.  Applying $\gamma$ to $\fhs = (1,1,2,3,2)$ we have 
\begin{eqnarray*}
\begin{array}{rccccc}
\mathbb{Z}_5& 0&1&2&3&4\\
\hline
{\cal F}^\gamma & 3&1&2&2&1
\end{array},
\end{eqnarray*}
with resulting FHS $(3,1,2,2,1)$ and corresponding $\df$ $(\mathbb{Z}_5; \{0\},\{1,4\},\{2,3\})$.  We observe that $1^\gamma=5$, so that $1^\gamma-1=4$.
Alternatively, we note that 
$\gamma^{-1}\rho_5\gamma=\begin{pmatrix}1&3&5&2&4\end{pmatrix}=\rho_5^2$. Hence $\phi_\gamma$ gives 
\begin{eqnarray*}
\begin{array}{rccccc}
\phi_\gamma(\mathbb{Z}_5)& 0&2&4&1&3\\
\hline
{\cal F} & 1&1&2&3&2
\end{array},
\end{eqnarray*}
which we can rewrite in order as
\begin{eqnarray*}
\begin{array}{rccccc}
\phi_\gamma(\mathbb{Z}_5)& 0&1&2&3&4\\
\hline
{\cal F} & 1&3&1&2&2
\end{array}.
\end{eqnarray*}
The resulting FHS is $(1,3,1,2,2)$, which is simply a cyclic shift of the one obtained previously.  The $\df$ is $(\mathbb{Z}_5; \{1\},\{0,2\},\{3,4\})$.  If we add $4$ to each element, we recover the previous $\df$.
\end{example}

\section{Conclusion}
\label{sec:discussion}
We have given a general definition of a disjoint difference family, and have seen a range of examples of applications in communications and
information security for these difference families, with different applications placing different constraints on the associated properties and parameters.  Focusing on the case of FHSs and their connection with partition type disjoint difference families, we have shown that a construction due to Fuji-Hara {\it et al.}\ \cite{FHS} gives rise to precisely the same disjoint difference families as an earlier construction of Lempel and Greenberger \cite{lempelgreenberger}, thus answering an open question in \cite{FHS}.  In response to the question of Fuji-Hara {\it et al.}\ as to when two FHSs can be considered to be ``essentially the same'' we have established a notion of equivalence of frequency hopping schemes.  FHSs based on a single sequence correspond to partition type disjoint difference families, and in this case we have shown that our definition of equivalence corresponds to an established notion of equivalence for difference families, although our definition also applies more generally to schemes based on more than one sequence.

\end{document}